\documentclass{article}

\usepackage{amssymb,amsmath,enumerate,amsthm,pseudocode,graphicx}
\usepackage[utf8]{inputenc}
\usepackage[english]{babel}
\newtheorem{theorem}{Theorem}[section]

\newtheorem{lemma}[theorem]{Lemma}

\title{A new quadratic-time number-theoretic algorithm to solve matrix multiplication problem}

\author{
Shrohan Mohapatra \\
School of Electrical Sciences \\
Indian Institute of Technology Bhubaneswar \\
Email ID: sm32@iitbbs.ac.in
}

\begin{document}

\maketitle

\abstract{
There have been several algorithms designed to optimise matrix multiplication. From schoolbook method with complexity $O(n^3)$ to advanced tensor-based tools with time complexity $O(n^{2.3728639})$ (lowest possible bound achieved), a lot of work has been done to reduce the steps used in the recursive version. Some group-theoretic and computer algebraic estimations also conjecture the existence of an $O(n^2)$ algorithm. This article discusses a quadratic-time number-theoretic approach that converts large vectors in the operands to a single large entity and combines them to make the dot-product. For two $n \times n$ matrices, this dot-product is iteratively used for each such vector. Preprocessing and computation makes it a quadratic time algorithm with a considerable constant of proportionality. Special strategies for integers, floating point numbers and complex numbers are also discussed, with a theoretical estimation of time and space complexity.
}

\section{Introduction}

The naive algorithm, i.e. the direct implementation of the matrix multiplication, with $\omega = 3$ grows slower rapidly with increase in size, which was for long thought to be the fastest. The first successful attempt in the sub-cubic domain was by Strassen\cite{strassen_first}. This approach attacked the problem by a divide-and-conquer strategy using the special algebraic identities, and has been commercially well-appreciated to be asymptotically faster than the normal method. This led to a plethora of subsequent attempts to discover tensor-based methods of reducing the number of recursive invocations at a fundamental level, leading Raz to a solution in NC \cite{raz}, Coppersmith and Winograd \cite{coppersmith} to a solution based on Sch\"onhage's theorem\cite{schonhage} and Salem-Spencer theorem on arithmetic progression\cite{salemSpencer}, and a lot more \cite{williams,kleinberg,pan}. The lowest ever bound achieved till now is $\omega = 2.3728639$ \cite{legall}. Despite all of this, these methods are difficult to implement on a digital computer. Post Strassen's algorithm design, all of the work started with unique algebraic identities in trilinear form, and then is taken to highly scaled tensor power, in order to reduce that identity to the general matrix multiplication problem. When the identity is taken to such high tensor powers, there is a requirement of large amount of memory and subsequently large number of additions and subtractions. Also partial multiplication algorithms render some errors as well \cite{williams,kleinberg,schonhage}, and recovering them adds to the complexity of both the algorithm as well as the manual sketching of its computable model. This results in the consequence that all of these remain just as sketches; all work has apparently been done to \textit{show the existence of an algorithm with the corresponding exponent $\omega$}.
\\ \\
This article focusses on development of an algorithm that uses some simple identities from number theory and the concept of convolution to accelerate the calculation of the dot product in constant time and consequently reduce the overall time complexity. The following explains the organisation of the rest of the paper. In section \ref{proof1}, I present some prerequisites to interface with my presentation of the algorithm, namely some important algebraic forms that help in matrix multiplication and some lemmas from number theory. Then follows section \ref{ProposedAlgorithm} where the algorithm has been described along with apt examples of the working, individually for positive integers, negative integers, floating point numbers and complex numbers. In section \ref{experiment}, some interesting experimental results based on time complexity and space complexity have been shown.

\section{Background} \label{proof1}

In this section, I present a follow-up of the bilinear and trilinear forms for total and partial matrix multiplications, followed by Sch\"onhage theorem. What follows are some trivial number-theoretic lemmas, based on the sum and product of numbers, which form the substance of proving the space complexity of the algorithm described in section \ref{pseudocode7}. 

\subsection{Some theorems on bilinear and trilinear algorithms} \label{schonhage}

From \cite{pan2}, the following defines the representation of bilinear algorithm for $(m,n,p)$ matrix product and then the theorem gives a bound on the matrix multiplication.

\begin{theorem} \label{T1}
Given two matrices, $A = [a_{ij}]_{m \times n}$, $B = [b_{ij}]_{n \times p}$, if there exists functions $f,f^{'},f^{''}$ for which the equation,
\begin{equation*}
\sum_{j} a_{ij} b_{jk} = \sum_{q=0}^{M-1} f^{''}(k,i,q) \sum_{j}f(i,j,q) a_{ij} \sum_{j}f^{'}(j,k,q) b_{jk} 
\end{equation*}
becomes an identity, the time complexity of the matrix multiplication is $O(N^\omega)$ where,
\begin{equation*}
\omega \leq 3\frac{log(M)}{log(mnp)}
\end{equation*}
\end{theorem}

Similarly, from \cite{pan2}, the following defines the trilinear algorithm for the matrix multiplication and provides the same bound. Both theorems \ref{T1} and \ref{T2} are proven \cite{cohn}.

\begin{theorem} \label{T2}
Given three matrices, $A = [a_{ij}]_{m \times n}$, $B = [b_{ij}]_{n \times p}$, $Z = [z_{ij}]_{p \times n}$, if there exists functions $f,f^{'},f^{''}$ for which the equation,
\begin{equation*}
\sum_{k,i}\sum_{j} a_{ij} b_{jk} z_{ki} = \sum_{q=0}^{M-1} \sum_{k,i} f^{''}(k,i,q)z_{ki} \sum_{j}f(i,j,q) a_{ij} \sum_{j}f^{'}(j,k,q) b_{jk} 
\end{equation*}
becomes an identity, the time complexity of the matrix multiplication is $O(N^\omega)$ where,
\begin{equation*}
\omega \leq 3\frac{log(M)}{log(mnp)}
\end{equation*}
\end{theorem}

Strassen's algorithm was originally presented as a normal recursion \cite{strassen_first}, and then represented in the form shown in theorems \ref{T1} and \ref{T2} \cite{stothers}. Majority of the subsequent prevalent approaches to this problem start with the Sch\"onhage's theorem \cite{williams,kleinberg,coppersmith}, shown below, which has also been proven \cite{schonhage}. This defines $\lambda$-trilinear form of the matrix multiplication.

\begin{theorem}
Assume given a field $F$, coefficients $\alpha_{i,j,h,l}$, $\beta_{j,k,h,l}$, $\gamma_{k,i,h,l}$ in $F(\lambda)$ and polynomials $f_g$ over $F$ such that
\begin{equation*}
\sum_{l=1}^{L} \sum_{i,j,h} \alpha_{i,j,h,l} x_{ij}^{[h]} \sum_{j,k,h} \beta_{i,j,h,l} y_{jk}^{[h]} \sum_{k,i,h} \gamma_{k,i,h,l} z_{ki}^{[h]} \\
= \sum_{h} \sum_{i=1}^{m_h} \sum_{j=1}^{n_h} \sum_{k=1}^{p_h} x_{ij}^{]h]} y_{jk}^{[h]} z_{ki}^{[h]} + 
\\ \sum_{g > 0} \lambda^g f_g(x_{ij}^{[h]},y_{jk}^{[h]},z_{ki}^{[h]})
\end{equation*}
is an identity in $x_{ij}^{[h]},y_{jk}^{[h]},z_{ki}^{[h]},\lambda$. Then given $\epsilon > 0$, one can construct an algorithm to multiply $N \times N$ matrices in $O(N^{3\tau+\epsilon})$ operations where $\tau$ satisfies
\begin{equation*}
L = \sum_{h} (m_h n_h p_h)^{\tau}
\end{equation*}
\end{theorem}

The above theorem is an approximate representation of several independent matrix multiplications, of dimensions $m_h \times n_h$ times $n_h \times p_h$, as a part of $L$ bilinear multiplications, where "$[h]$" is the superscript of the matrix element involved in the product \cite{coppersmith}. The approximation is rendered in a more refined way as follows:

\begin{equation*}
\sum_{l=1}^{L} \sum_{i,j,h} \alpha_{i,j,h,l} x_{ij}^{[h]} \sum_{j,k,h} \beta_{i,j,h,l} y_{jk}^{[h]} \sum_{k,i,h} \gamma_{k,i,h,l} z_{ki}^{[h]} \\
= \sum_{h} \sum_{i=1}^{m_h} \sum_{j=1}^{n_h} \sum_{k=1}^{p_h} x_{ij}^{]h]} y_{jk}^{[h]} z_{ki}^{[h]} + O(\lambda)
\end{equation*}

Previous works on reduction of the complexity \cite{williams,kleinberg,coppersmith} post Strassen's algorithm \cite{strassen_first} (that runs in $O(n^{\log_{2}7})$) start with some special identity in this $\lambda$-trilinear form and then take an asymptotically high tensor power of the same. As is obvious, the basic bilinear and trilinear forms embed in themselves a visibly implementable algorithm. But algorithms in $\lambda$-trilinear form are hard in terms of digital computation merely due to high tensor powers. In section \ref{proposal}, I shall present an algebraic number-theoretic identity in bilinear form that is visibly a transformation of matrix multiplication. I shall not be proving the existence of an algorithm using theorem \ref{T2}, rather would be presenting the pseudocode directly for different cases in section \ref{ProposedAlgorithm}.

\subsection{Some number-theoretic proofs} \label{proof2}

\begin{lemma} \label{Ref1}
Largest number of digits in the product of two $m$-digit numbers is $2m$, where $m \in \mathbb{N}$.
\end{lemma}

\begin{proof}
Largest possible $m$-digit number is (10$^m$-1). The largest possible product $P$ of two $m$-digit numbers is, thus,
\begin{equation*}
(10^m-1)^2 = (10^{2m} - 2 \cdot 10^m + 1)
\end{equation*}
The number of digits in $P$ is
\begin{equation*}
\lceil\log_{10}P\rceil =2m
\end{equation*}
\end{proof}

\begin{lemma} \label{Ref2}
If the number of digits in the sum of $q$ $m$-digit numbers is $N$ then $N < m + \lceil\log_{10}q\rceil$, where $q,m \in \mathbb{N}$.
\end{lemma}

\begin{proof}
(10$^m$-1) being the largest $m$-digit number,
\begin{equation*}
N = \lceil \log_{10}(10^m-1)\cdot q \rceil
\\ < \lceil \log_{10}(q \cdot 10^m) \rceil
\\ = m + \lceil\log_{10}q\rceil
\end{equation*}
\end{proof}

I shall be using these lemmas while anchoring some parameters in the identity presented in section \ref{proposal}, and also while proving the space complexity of my proposed algorithm for the case of non-negative integers in section \ref{pseudocode7}.

\subsection{My proposed identity in bilinear form} \label{proposal}

The following identity based on number-theoretic algebra forms the foundation of the algorithm design, shown in section \ref{pseudocode7}. The choice of integer $P$ should be a practical one. Any decimal number $n_d$ multiplied by an exponent $r$ of 10 pads $r$ zeroes to the right of $n_d$. This padding is guided by lemmas \ref{Ref1} and \ref{Ref2}, so that the matrix product is a part of the overall block computation, similar to the Sch\"onhage's theorem shown in section \ref{schonhage}.

\begin{lemma} \label{Ref3}
Given two matrices $A,B \in \{\mathbb{N} \cup \{0\}\}^{n \times n}$ and an integer $P$,
\begin{equation*}
(AB)_{i,j} = \sum_{k=0}^{n-1} A[i][k] B[k][j] = \bigg\lfloor \frac{\sum_{k=0}^{n-1} A[i][k]10^{(n-1-k)P} \sum_{k=0}^{n-1} B[k][j]10^{kP}}{10^{(n-1)P}} \bigg\rfloor \mod 10^{P}
\end{equation*}
\end{lemma}

\begin{proof}
\begin{gather*}
\sum_{k=0}^{n-1} A[i][k]10^{(n-1-k)P} \sum_{k=0}^{n-1} B[k][j]10^{kP} = \\
\sum_{k_1=0}^{n-1} \sum_{k_2=0}^{n-1} A[i][k_1] \cdot B[k_2][j] 10^{(n-1-k_1+k_2)P} = \\
10^{2(n-1)P} \cdot A[i][0] \cdot B[n-1][j] + 10^{(2n-3)P} \sum_{k=0}^{1} A[i][k] B[k+n-2][j] + \\ \cdots 
+ 10^{nP} \sum_{k=0}^{n-2} A[i][k] B[k+1][j] + 10^{(n-1)P} \sum_{k=0}^{n-1} A[i][k] B[k][j] + \\ 10^{(n-2)P}\sum_{k=1}^{n-1} A[i][k] B[k-1][j] + \cdots + A[i][n-1] B[0][j] \\
\implies \bigg\lfloor \frac{\sum_{k=0}^{n-1} A[i][k]10^{(n-1-k)P} \sum_{k=0}^{n-1} B[k][j]10^{kP}}{10^{(n-1)P}} \bigg\rfloor = \\
10^{(n-1)P} \cdot A[i][0] \cdot B[n-1][j] + 10^{(n-2)P} \sum_{k=0}^{1} A[i][k] B[k+n-2][j] \\ + \cdots 
+ 10^{P} \sum_{k=0}^{n-2} A[i][k] B[k+1][j] + \sum_{k=0}^{n-1} A[i][k] B[k][j] \\
\implies \bigg\lfloor \frac{\sum_{k=0}^{n-1} A[i][k]10^{(n-1-k)P} \sum_{k=0}^{n-1} B[k][j]10^{kP}}{10^{(n-1)P}} \bigg\rfloor \mod 10^{P} = \sum_{k=0}^{n-1} A[i][k] B[k][j]
\end{gather*}
\end{proof}

\section{My proposed algorithm} \label{ProposedAlgorithm}

In this section, I first present my algorithm for the non-negative integers, followed by a working example. Then small modifications in the mechanism are presented to deal with the negative integers, floating point numbers and complex numbers.

\subsection{Pseudocode for the case of the non-negative integers} \label{pseudocode7}

The problem for the case of positive integers is being approached in the following way: first the rows of the pre-multiplicand and the columns of the post-multiplicand by joining the apt number of zeroes, a parameter $P$ that is dependent on the theoretical background set up on lemmas 2.3 and 2.4. This compresses the operand matrices into a $n \times 1$ matrix and $1 \times n$ matrix, whose product will be further broken down. The entire implementation is completely based on the identity I have proposed in section \ref{proposal}.

\begin{pseudocode}{matrixMultiplyPOSITIVEinteger}{A[[0,0]...[n-1,n-1]],B[[0,0]...[n-1,n-1]]}
m \GETS 0 \\
C \GETS \{0,0,0 ... (n$ $times) ... 0\} \\
D \GETS \{0,0,0 ... (n$ $times) ... 0\} \\
E \GETS \{\{0,0,0 ... (n$ $times) ... 0\}, \{0,0,0 ... (n$ $times) ... 0\}, .... (n$ $times) .... \{0,0,0 ... (n$ $times) ... 0\} \} \\
\FOR i \GETS 0 \TO n-1 \DO
	\BEGIN
	\FOR j \GETS 0 \TO n-1 \DO
		\BEGIN
		\IF m < A[i][j] \THEN m \GETS A[i][j] \\
		\IF m < B[i][j] \THEN m \GETS B[i][j] \\
		\END \\
	\END \\
M \GETS \lceil \log_{10}m \rceil \\
P \GETS \lceil \log_{10}n(10^{2M}-1) \rceil \\
\FOR i \GETS 0 \TO n-1 \DO
	\BEGIN
	\FOR j \GETS 0 \TO n-1 \DO
		C[i] \GETS C[i] \cdot 10^P + A[i][j]
	\END \\
\FOR j \GETS 0 \TO n-1 \DO
	\BEGIN
	\FOR i \GETS 0 \TO n-1 \DO
		D[j] \GETS D[j] \cdot 10^P + B[n-1-i][j]
	\END \\
\FOR i \GETS 0 \TO n-1 \DO
	\BEGIN
	\FOR j \GETS 0 \TO n-1 \DO
	E[i][j] \GETS  \lfloor \frac{C[i] \cdot D[j]}{10^{P(n-1)}} \rfloor \mod 10^P
	\END \\
\RETURN{E}\\
\end{pseudocode}

Clearly since most of the programmatically atomic instructions in the mentioned algorithm are iterated $n^2$ times, $n$ being the size of the operands, the time complexity is $O(n^2)$. The majority of the space of the memory occupied as a consequence of the pseudocode is by the arrays $C$, $D$ and $E$. Each element of array $C$ requires $n \times P$ number of decimal digits, where $P = \lceil \log_{10}n(10^{2M}-1) \rceil$ and $M$ is the maximum number of digits of the elements of the operands. Following the lines of lemma \ref{Ref2} the total size of the array $C$ as well as $D$ is $2n^2P = 2 n^2 \lceil \log_{10}n(10^{2M}-1) \rceil = O\bigg(n^2 log(n)\bigg)$. For the computation of each element of $E$ is the product of two elements of $C$ and $D$, which require $nP$ decimal digits. From lemma \ref{Ref1}, each element of $E$ requires $2nP$ decimal digits. So in total, the array occupies $2n^3P < 2 n^3 (2M+\log_{10}n)$. Thus, the overall space complexity is $O\bigg(n^3 log(n)\bigg)$. 

\subsection{Example of the working of the algorithm} \label{example}

Let us assume the following two $3 \times 3$ matrices for the positive integer matrix multiplication algorithm.

\begin{equation*}
A =
  \begin{bmatrix}
    1 & 2 & 3 \\
    4 & 5 & 6 \\
    7 & 8 & 9
  \end{bmatrix}, \\
B =
  \begin{bmatrix}
    9 & 8 & 7 \\
    6 & 5 & 4 \\
    3 & 2 & 1
  \end{bmatrix}
\end{equation*}

The values of $M$ and $P$ are 1 and 3 respectively. Consequently matrices $C$ and $D$ of the algorithm are as follows.

\begin{equation*}
C = \begin{bmatrix}
		1002003 & 4005006 & 7008009
	\end{bmatrix}
\end{equation*}

\begin{equation*}
D = \begin{bmatrix}
		3006009 & 2005008 & 1004007
	\end{bmatrix}
\end{equation*}

The product $C[0] \times D[0]$ contains the element $E[0][0]$, which has been shown in the bold below.

\begin{equation*}
C[0] \times D[0] = 1002003 \times 3006009 = 3,012,\textbf{030},036,027
\end{equation*}

Similarly, following the algorithm, the resultant matrix $E$ is given as the one below.

\begin{equation*}
E =
  \begin{bmatrix}
    30 & 24 & 18 \\ 
    84 & 69 & 54 \\
    138 & 114 & 90
  \end{bmatrix}
\end{equation*}

\subsection{Handling the case of the negative integers} \label{Negative}

Here the idea is to separate every single negative integer into two non-negative integers; trivially, for all $m>0$, $(-m) = 0 - m$. So for each of the operand matrices $A$, there are two such non-negative matrices $A_1$ and $A_2$, such that,

\begin{equation*}
A = A_1 - A_2
\end{equation*}

where

\begin{equation*}
A_1[i,j] = \Bigg\{
\begin{array}{ll}
      0 & A[i][j] < 0 \\
      A[i][j] & A[i][j] \geq 0\\
\end{array} \\ \\
A_2[i,j] = \Bigg\{
\begin{array}{ll}
      -(A[i][j]) & A[i][j] < 0 \\
      0 & A[i][j] \geq 0\\
\end{array}
\end{equation*}

For example, consider the following matrix and its decomposition.

\begin{equation*}
  \begin{bmatrix}
    3 & -2 & 1 \\ 
    -8 & 6 & 5 \\
    18 & -14 & -9
  \end{bmatrix} = \\
  \begin{bmatrix}
    3 & 0 & 1 \\
    0 & 6 & 5 \\
    18 & 0 & 0
  \end{bmatrix} - \\
  \begin{bmatrix}
    0 & 2 & 0 \\
    8 & 0 & 0 \\
    0 & 14 & 9
  \end{bmatrix}
\end{equation*}

So the product of the matrices $A = A_1 - A_2$ and $B = B_1 - B_2$ containing integers, I decompose them in the following way.

\begin{equation*}
AB = (A_1 - A_2)(B_1 - B_2) = A_1 B_1 - A_2 B_1 - A_1 B_2 + A_2 B_2
\end{equation*}

As is shown above, the matrix multiplication for integers, in general, can be decomposed into four instances of the previous case for positive integers. The pseudocode below for the method matrixMultiplyGENERALinteger() is self-explanatory to the approach discussed above.

\begin{pseudocode}{matrixMultiplyGENERALinteger}{A[[0,0]...[n-1,n-1]],B[[0,0]...[n-1,n-1]]}
m \GETS 0 \\
A1 \GETS \{\{0,0,0 ... (n$ $times) ... 0\}, \{0,0,0 ... (n$ $times) ... 0\}, .... (n$ $times) .... \{0,0,0 ... (n$ $times) ... 0\} \} \\
A2 \GETS \{\{0,0,0 ... (n$ $times) ... 0\}, \{0,0,0 ... (n$ $times) ... 0\}, .... (n$ $times) .... \{0,0,0 ... (n$ $times) ... 0\} \} \\
B1 \GETS \{\{0,0,0 ... (n$ $times) ... 0\}, \{0,0,0 ... (n$ $times) ... 0\}, .... (n$ $times) .... \{0,0,0 ... (n$ $times) ... 0\} \} \\
B2 \GETS \{\{0,0,0 ... (n$ $times) ... 0\}, \{0,0,0 ... (n$ $times) ... 0\}, .... (n$ $times) .... \{0,0,0 ... (n$ $times) ... 0\} \} \\
\FOR i \GETS 0 \TO n-1 \DO
	\BEGIN
	\FOR j \GETS 0 \TO n-1 \DO
		\BEGIN
		\IF A[i][j] \geq 0
		\THEN
		\BEGIN
		A1[i][j] \GETS A[i][j] \\ A2[i][j] \GETS 0
		\END
		\ELSE
		\BEGIN
		A1[i][j] \GETS 0 \\ A2[i][j] \GETS -A[i][j]
		\END \\
		\IF B[i][j] \geq 0
		\THEN
		\BEGIN
		B1[i][j] \GETS B[i][j] \\ B2[i][j] \GETS 0
		\END
		\ELSE
		\BEGIN
		B1[i][j] \GETS 0 \\ B2[i][j] \GETS -B[i][j]
		\END \\
		\END \\
	\END \\
C1 \GETS matrixMultiplyPOSITIVEinteger(A1,B1) \\
C2 \GETS matrixMultiplyPOSITIVEinteger(A1,B2) \\
C3 \GETS matrixMultiplyPOSITIVEinteger(A2,B1) \\
C4 \GETS matrixMultiplyPOSITIVEinteger(A2,B2) \\
C \GETS \{\{0,0,0 ... (n$ $times) ... 0\}, \{0,0,0 ... (n$ $times) ... 0\}, .... (n$ $times) .... \{0,0,0 ... (n$ $times) ... 0\} \} \\
\FOR i \GETS 0 \TO n-1 \DO
	\BEGIN
	\FOR j \GETS 0 \TO n-1 \DO
		C[i][j] \GETS C1[i][j] - C2[i][j] - C3[i][j] + C4[i][j]
	\END \\
\RETURN{C}\\
\end{pseudocode}

Clearly, other than the four invocations of the method matrixMultiplyPOSITIVEinteger(), all require $O(n^2)$ iterations of programmatically atomic instructions. Thus, time complexity of the matrix multiplication still remains as $O(n^2)$. Similarly, additional two-dimensional matrices do not affect the space complexity. And as will be visible in sections \ref{FloatingPoint} and \ref{ComplexMatrix}, the paradigm adopted still runs in quadratic time and does not affect time and space complexity.

\subsection{Handling the case of floating point numbers} \label{FloatingPoint}

Here the case of non-negative floating point numbers is shown. One can use the following guidelines to take care of the negative ones using a strategy similar to the one discussed in section \ref{Negative}. The basic rationale is to

\begin{enumerate}
	\item Find the maximum number of digits $R_1,R_2$ to the right of the decimal point for each of the operands.
	\item Multiply each of them with scalars obtained by raising 10 to the exponents $R_1$ and $R_2$ respectively, to transform the operands into non-negative integer matrices.
	\item Use the method matrixMultiplyPOSITIVEinteger() explained in section \ref{pseudocode7} with the transformed operands.
	\item Divide the resultant matrix by a scalar $10^{R_1+R_2}$.
\end{enumerate}

This strategy has been well-illustrated in the pseudocode for the method matrixMultiplyFLOATpoint() below, that uses a trivial method numberOfDecimalDigits() that calculates the number of digits to the right of the decimal point for a given floating point number. During it's implementation on a microprocessor, one needs to be careful in handling the data types.

\begin{pseudocode}{numberofDecimalDigits}{x}
	\RETURN{-(\lfloor \log_{10}(x-\lfloor x \rfloor) \rfloor+1}
\end{pseudocode}

\begin{pseudocode}{matrixMultiplyFLOATpoint}{A[[0,0]...[n-1,n-1]],B[[0,0]...[n-1,n-1]]}
R_1,R_2 \GETS 0,0 \\
A1 \GETS \{\{0,0,0 ... (n$ $times) ... 0\}, \{0,0,0 ... (n$ $times) ... 0\}, .... (n$ $times) .... \{0,0,0 ... (n$ $times) ... 0\} \} \\
B1 \GETS \{\{0,0,0 ... (n$ $times) ... 0\}, \{0,0,0 ... (n$ $times) ... 0\}, .... (n$ $times) .... \{0,0,0 ... (n$ $times) ... 0\} \} \\
\FOR i \GETS 0 \TO n-1 \DO
	\BEGIN
	\FOR j \GETS 0 \TO n-1 \DO
		\BEGIN
		\IF R_1 < numberOfDecimalDigits(A[i][j])
		\THEN R_1 \GETS numberOfDecimalDigits(A[i][j]) \\ 
		\IF R_2 < numberOfDecimalDigits(B[i][j])
		\THEN R_2 \GETS numberOfDecimalDigits(B[i][j]) \\
		\END \\
	\END \\
\FOR i \GETS 0 \TO n-1 \DO
	\BEGIN
	\FOR j \GETS 0 \TO n-1 \DO
		\BEGIN
		A1[i][j] \GETS A[i][j] \cdot 10^{R_1} \\
		B1[i][j] \GETS B[i][j] \cdot 10^{R_2}
		\END \\
	\END \\
C \GETS matrixMultiplyPOSITIVEinteger(A1,B1)\\
\FOR i \GETS 0 \TO n-1 \DO
	\BEGIN
	\FOR j \GETS 0 \TO n-1 \DO
		C[i][j] \GETS \frac{C[i][j]}{10^{R_1+R_2}}
	\END \\
\RETURN{C}\\
\end{pseudocode}

\subsection{Handling the case of complex numbers} \label{ComplexMatrix}

This case is similar, and even simpler, than that of the negative numbers discussed in section \ref{Negative}. Any $n \times n$ matrix $A \in \mathbb{C}^{n \times n}$, can be broken into two matrices $A_r,A_i \in \mathbb{R}^{n \times n}$

\begin{equation*}
A = A_r + \iota A_i
\end{equation*}

where $\iota$ is the imaginary unit $\sqrt{-1}$. Thus, following the rules of complex algebra, the product of two complex matrices $A,B$ can be written as

\begin{equation*}
AB = (A_r + \iota A_i)(B_r + \iota B_i) = (A_r B_r - A_i B_i) + \iota (A_r B_i + A_i B_r)
\end{equation*}

Again, similar to the case of negative integers, problem gets reduced to four instances of matrix multiplication of real numbers, as is clearly visible in the pseudocode discussed below. This algorithm is implemented by the method matrixMultiplyCOMPLEXnumber() that uses method matrixMultiplyFLOATpoint() discussed in section \ref{FloatingPoint}.

\begin{pseudocode}{matrixMultiplyCOMPLEXnumber}{A[[0,0]...[n-1,n-1]],B[[0,0]...[n-1,n-1]]}
A_r \GETS Real(A) \\ \COMMENT{Real(A) method takes the real part of all the elements of A $\in \mathbb{C}^{n \times n}$} \\
B_r \GETS Real(B) \\
A_i \GETS Im(A)\\ \COMMENT{Im(A) method takes the imaginary part of all the elements of A $\in \mathbb{C}^{n \times n}$} \\
B_i \GETS Im(B)\\
C_1 \GETS matrixMultiplyFLOATpoint(A_r,B_r) \\
C_2 \GETS matrixMultiplyFLOATpoint(A_r,B_i) \\
C_3 \GETS matrixMultiplyFLOATpoint(A_i,B_r) \\
C_4 \GETS matrixMultiplyFLOATpoint(A_i,B_i) \\
C \GETS \{\{0,0,0 ... (n$ $times) ... 0\}, \{0,0,0 ... (n$ $times) ... 0\}, .... (n$ $times) .... \{0,0,0 ... (n$ $times) ... 0\} \} \\
\FOR i \GETS 0 \TO n-1 \DO
	\BEGIN
	\FOR j \GETS 0 \TO n-1 \DO
		C[i][j] \GETS (C_1[i][j]-C_4[i][j]) + \iota (C_2[i][j]+C_3[i][j])
	\END \\
\RETURN{C}\\
\end{pseudocode}

\section{Experimental results} \label{experiment}

The algorithm, described in method matrixMultiplyPOSITIVEinteger() of section \ref{pseudocode7}, has been compared with the IJK-algorithm and Strassen's algorithm in terms of time and space complexity. Other sophisticated algorithms with theoretically faster in terms of asymptotic complexity \cite{williams, coppersmith, cohn, legall, pan, stothers} are really hard and quixotic to implement in a digital computer, and some of them do not guarantee of complete accuracy too \cite{schonhage}. The implementation of the three algorithms are done in Python language, that uses efficient memory management using pymalloc allocator. Figures 1 and 2 show a comparison of space and time complexity respectively. The memory footprint (in bytes) is calculated very carefully considering only the arrays $C,D$ and $E$ in the pseudocode described in the section \ref{pseudocode7}. The justification to the trend in the time complexity is conspicuous, but that of the space complexity is not. The IJK method does not consume anything more than the output matrix, the Strassen method requires 24 intermediate matrices for the 25 steps required in every recursive call \cite{strassen_first} and our algorithm consumes $O\Big(n^3 log(n)\Big)$ space as shown in section \ref{pseudocode7}.

\begin{figure}
\includegraphics[clip, trim = 1.2cm 0.01cm 0.01cm 0.01cm, scale=0.75]{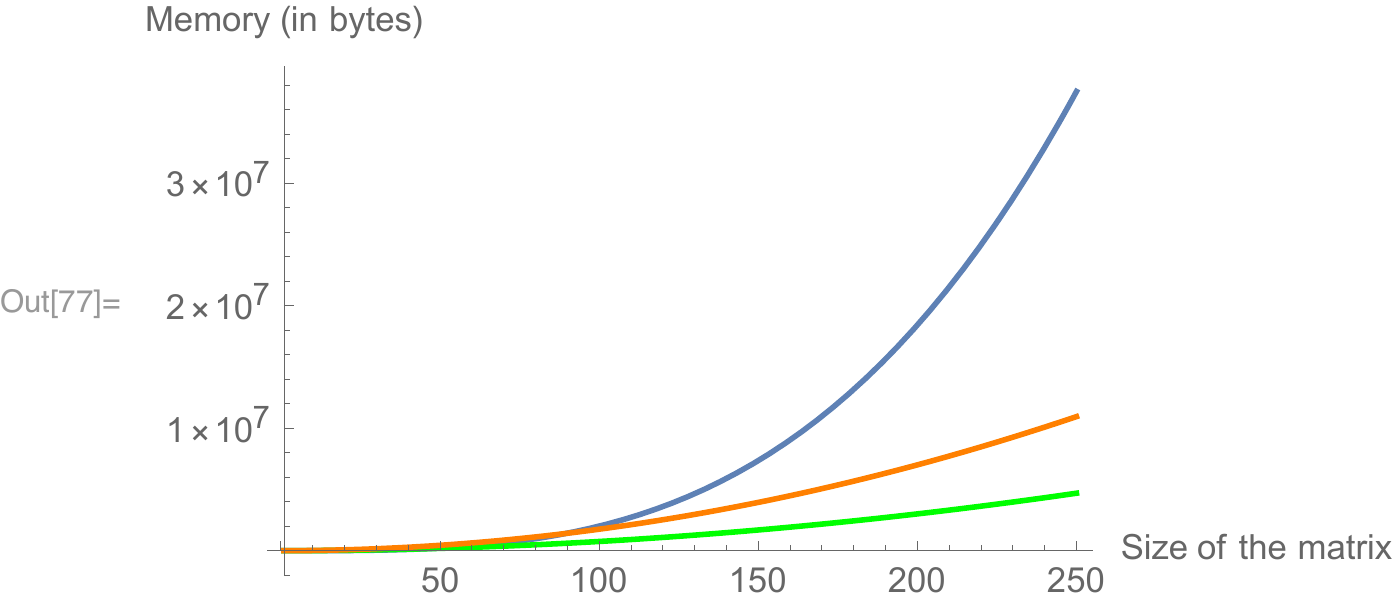}
\caption{A graph of memory footprint (in bytes) of the algorithms against the size of the input matrix. The blue curve represents the trend of the memory footprint for the proposed algorithm, the orange curve represents that of the Strassen's algorithm and the green curve represents that of the schoolbook algorithm.}
\end{figure}

\begin{figure}
\includegraphics[clip, trim = 1.2cm 0.01cm 0.01cm 0.01cm, scale=0.75]{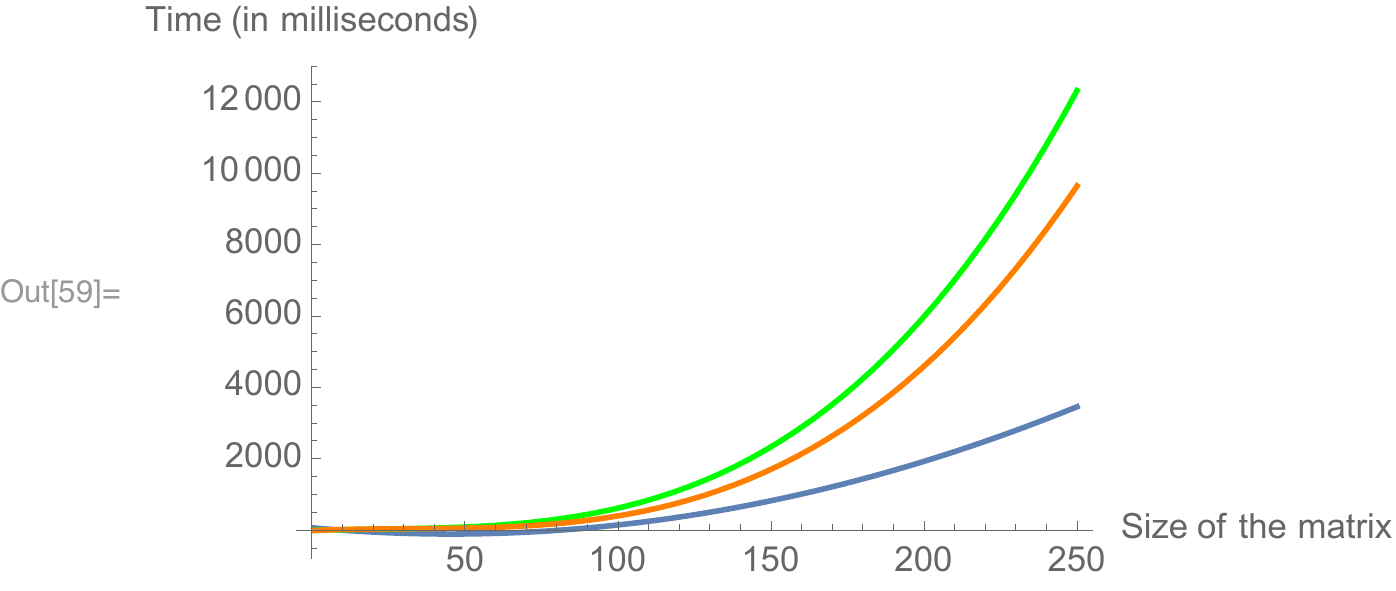}
\caption{A graph of time consumption (in milliseconds) of the algorithms against the size of the input matrix. The blue curve represents the trend of the time complexity for the proposed algorithm, the orange curve represents that of the Strassen's algorithm and the green curve represents that of the schoolbook algorithm.}
\end{figure}

\section{Conclusion} \label{conclusion}

This article begins with revisiting some important symbolic computations that form the basis of some prevalent sub-cubic matrix multiplication algorithms. Such computations reveal the complexity of the algorithm only in the asymptotic domain, and the implementation is highly involved and consequently will incur high overheads. Then following trivial number-theoretic foundations, the quadratic-time algorithm is presented first in the bilinear form and then in the form of a pseudocode. Due to the nature of our solution, separate cases for integers, reals and complex numbers are considered in the design of the algorithm. Finally, experimental results are presented that compare the algorithm with the prevalently used ones in terms of time and memory.

The paper attempts to close in upon the conjecture that the matrix multiplication exponent $\omega = 2$, and probably does more than that. The pseudocode reveals that the implementation would demand large amount of memory, but would definitely be faster than schoolbook and Strassen's algorithm. In the future, exploitation of parallel architecture, such as CPU-GPU interaction, multicore setup etc., can lead to a sub-quadratic design. Formalising the same requires the concept of cellular automata. Also, since our algorithm is completely based on number theoretic lemmas, this cannot be directly or simply relied upon in the case of symbolic computations. This, I believe, shall be a challenge from here on, as one can definitely come up with an algorithm using Sch\"onhage's theorem, but would again require very high computational power due to consequently high tensor powers and high degree of empiricism and approximation. An optimal quadratic algorithm for symbolic computation, according to me, should start from different axioms for it to be usable in a modern-day microprocessor.

\end{document}